\documentclass{article}
\usepackage{amsfonts,amssymb,amsmath,a4wide,paralist}
\usepackage{epsfig,rotating,framed,enumerate,listings}
\usepackage{color}
\usepackage{paralist}
\usepackage{marginnote,nicefrac}
\usepackage{authblk}
\usepackage{tikz}
\usepackage{xfrac}

\usepackage[paper=a4paper,left=31mm,right=31mm,top=30mm, bottom=25mm]{geometry}

\usepackage{hyperref}

\newcommand{\bigo}{\ensuremath{\mathcal{O}}}
\newcommand{\dpws} {\text{d-pw}}
\newcommand{\dtws} {\text{d-tw}}

\newcommand{\dtw} {\text{d-tw}}
\newcommand{\pws} {\text{pw}}
\newcommand{\tws} {\text{tw}}

\newcommand{\un} {{\it und}}

\newtheorem{theorem}{Theorem}[section]
\newtheorem{example}[theorem]{Example}
\newtheorem{lemma}[theorem]{Lemma}
\newtheorem{definition}[theorem]{Definition}

\newtheorem{corollary}[theorem]{Corollary}

\newtheorem{remark}[theorem]{Remark}

\newenvironment{proof}{\noindent{\bf Proof~}}{\null\hfill $\Box$\par\medskip}

\definecolor{light-gray}{gray}{0.5}

\begin{document}

\title{Computing directed path-width and directed tree-width of recursively defined digraphs\thanks{A short version of 
this paper appeared in Proceedings of the {\em International Computing and Combinatorics Conference} (COCOON 2018) \cite{GR18c}.}}

\author[1]{Frank Gurski}
\author[1]{Carolin Rehs}

\affil[1]{\small University of  D\"usseldorf,
Institute of Computer Science, Algorithmics for Hard Problems Group,\newline 
40225 D\"usseldorf, Germany}

\maketitle

\begin{abstract}
In this paper we consider the directed path-width and directed tree-width of recursively defined digraphs.
As an important  combinatorial tool, we show how the directed path-width and the directed tree-width
can be computed for the disjoint union, order composition, directed union, and series composition
of two directed graphs. These results imply the equality of 
directed path-width and directed tree-width for all digraphs which can be 
defined by these four operations. This allows us to show a linear-time solution for computing  
the directed path-width and directed tree-width of all these digraphs.
Since directed co-graphs are precisely those digraphs which can be defined by 
the disjoint union, order composition,  and series composition
our results imply the equality of 
directed path-width and directed tree-width for directed co-graphs
and also a linear-time solution for computing  
the directed path-width and directed tree-width of directed co-graphs,
which generalizes the known results for undirected co-graphs
of Bodlaender and M\"ohring.

\end{abstract}

\section{Introduction}

Tree-width is a well-known graph parameter \cite{RS86}. 
Many NP-hard graph problems admit poly\-nomial-time solutions when restricted to 
graphs of bounded tree-width using the tree-de\-com\-po\-sition \cite{Arn85,AP89,Hag00,KZN00}. 
The same holds for path-width \cite{RS83} since a path-decomposition 
can be regarded as a special case of a tree-decomposition. 
Computing both parameters is hard even for bipartite graphs and complements
of bipartite graphs \cite{ACP87}, while for
co-graphs it has been shown \cite{BM90,BM93} that the path-width equals the 
tree-width and how to compute this value in linear time.

During the last years, width parameters for directed graphs 
have received a lot of attention \cite{GHKMORS16}.
Among these are directed path-width and directed tree-width
\cite{JRST01}. 
Since for complete bioriented digraphs the directed path-width equals the
(undirected) path-width of the corresponding underlying undirected graph
it follows that
determining whether the directed path-width of some given digraph  is 
at most some given value $w$ is NP-complete. The same holds for directed tree-width. 
There is an XP-algorithm for directed path-width w.r.t.~the standard parameter  
by \cite{KKKTT16}, which and implies that for each constant $w$, it is 
decidable in polynomial time whether a given 
digraph has directed path-width at most $w$. The same holds for
directed tree-width by \cite{JRST01}. 
This motivates to consider the recognition
problem restricted to special digraph classes.

We show useful properties of directed path-de\-com\-positions and
directed tree-de\-com\-positions, such as bidirectional complete subdigraph
and bidirectional complete bipartite subdigraph lemmas.
These results allow us to show how the directed path-width and directed tree-width
can be computed for the disjoint union, order composition, directed union, and
series composition of two directed graphs. Our proofs are constructive, i.e.\
a directed path-decomposition and a directed tree-decomposition 
can be computed from a given expression. These results imply the equality of 
directed path-width and directed tree-width for all digraphs which can be 
defined by the disjoint union, order composition, directed union, and
series composition. This allows us to show a linear-time solution for computing  
the directed path-width and directed tree-width of all these digraphs.
Among these are directed 
co-graphs, which can be defined by disjoint union, order composition, and
series composition \cite{CP06}. Directed co-graphs are useful to characterize
digraphs of directed NLC-width 1 and digraphs of directed clique-width 2 \cite{GWY16} 
and are useful for the reconstruction of the evolutionary 
history of genes or species using genomic sequence data \cite{HSW17,NEMWH18}.
Our results imply the equality of 
directed path-width and directed tree-width for directed co-graphs
and a linear-time solution for computing  
the directed path-width and directed tree-width of directed co-graphs.
Since for complete bioriented digraphs the directed path-width equals the
(undirected) path-width of the corresponding underlying undirected graph
and the directed tree-width equals the
(undirected) tree-width of the corresponding underlying undirected graph
our results generalize the known results from \cite{BM90,BM93}.

\section{Preliminaries}\label{sec-prel}

We use the notations of Bang-Jensen and Gutin \cite{BG09} for graphs and digraphs.

\subsection{Graphs}

A {\em graph} is a pair  $G=(V,E)$, where $V$ is 
a finite set of {\em vertices} 
and $E \subseteq \{ \{u,v\} \mid u,v \in
V,~u \not= v\}$ is a finite set of {\em edges}.
A graph $G'=(V',E')$ is a {\em subgraph} of graph $G=(V,E)$ if $V'\subseteq V$ 
and $E'\subseteq E$.  If every edge of $E$ with both end vertices in $V'$  is in
$E'$, we say that $G'$ is an {\em induced subgraph} of digraph $G$ and 
we write $G'=G[V']$.
For some undirected graph $G=(V,E)$ its complement graph is defined by
$$\overline{G}=(V,\{\{u,v\}~|~\{u,v\}\not\in E, u,v\in V, u\neq v\}).$$

\subsection{Recursively defined Graphs}

\subsubsection{Operations}

Let $G_1=(V_1,E_1), \ldots, G_k=(V_k,E_k)$ be $k$ vertex-disjoint graphs. 
\begin{itemize}
\item
The {\em disjoint union} of $G_1, \ldots, G_k$, 
denoted by $G_1 \cup \ldots \cup G_k$, 
is the graph with vertex set $V_1\cup \ldots \cup V_k$ and 
edge set $E_1\cup \ldots \cup E_k$. 

\item
The {\em join composition} of $G_1, \ldots, G_k$, 
denoted by $G_1\times \ldots \times G_k$, 
is defined by their disjoint union plus all possible edges between
vertices of $G_i$ and $G_j$  for all $1\leq i,j\leq k$, $i\neq j$.
\end{itemize}

\subsubsection{Co-graphs}

Co-graphs have been introduced in the 1970s by a number of authors
under different notations, such as hereditary Dacey graphs (HD graphs) in \cite{Sum74},
$D^*$-graphs in \cite{Jun78}, 2-parity graphs in \cite{BU84}, and complement reducible
graphs (co-graphs) in \cite{Ler71}. Co-graphs can be characterized 
as the set of graphs without an induced path with four vertices \cite{CLS81}.
From an algorithmic point of view the following recursive 
definition  is very useful.

\begin{definition}[Co-graphs]
The class of {\em co-graphs} is recursively defined as follows.
\begin{enumerate}[(i)]
\item Every graph on a single vertex $(\{v\},\emptyset)$, 
denoted by $\bullet$, is a  {\em co-graph}.

\item If  $G_1,\ldots,G_k$  are  vertex-disjoint co-graphs, then 
\begin{enumerate}
\item 
the disjoint union  
$G_1\cup\ldots \cup G_k$ and 

\item
the join composition 
$G_1\times \ldots \times G_k$  are  {\em co-graphs}.
\end{enumerate}
\end{enumerate}
\end{definition}

By this definition every co-graph can be represented by a tree structure, 
denoted as {\em co-tree}. The leaves of the co-tree represent the 
vertices of the graph and the inner nodes of the co-tree  correspond 
to the operations applied on the
subexpressions defined by the subtrees.  
For every graph $G$ one can decide in linear time, whether $G$ is
a co-graph and in the case of a positive answer construct a co-tree
for $G$, see \cite{HP05}. 
Using the co-tree a lot of hard problems have been shown to be 
solvable in polynomial time when restricted to co-graphs.
Such problems are clique, independent set, partition into
independent sets (chromatic number), partition into cliques, hamiltonian cycle, 
isomorphism \cite{CLS81}.

\subsection{Digraphs}

A {\em directed graph} or {\em digraph} is a pair  $G=(V,E)$, where $V$ is 
a finite set of {\em vertices} and  $E\subseteq \{(u,v) \mid u,v \in
V,~u \not= v\}$ is a finite set of ordered pairs of distinct
vertices called {\em arcs}. 
A digraph $G'=(V',E')$ is a {\em subdigraph} of digraph $G=(V,E)$ if $V'\subseteq V$ 
and $E'\subseteq E$.  If every arc of $E$ with both end vertices in $V'$  is in
$E'$, we say that $G'$ is an {\em induced subdigraph} of digraph $G$ and we 
write $G'=G[V']$.
For some digraph $G=(V,E)$ its {\em complement digraph} is defined by
$$\overline{G}=(V,\{(u,v)~|~(u,v)\not\in E, u,v\in V, u\neq v\})$$
and its {\em converse digraph} is defined by
$$G^c=(V,\{(u,v)~|~(v,u)\in E, u,v\in V, u\neq v\}).$$

Let $G=(V,E)$ be a digraph.
\begin{itemize}
\item $G$ is {\em edgeless} if for all $u,v \in V$, $u \neq v$, 
none of the two pairs $(u,v)$ and $(v,u)$ belongs to $E$.

\item $G$ is a {\em tournament} if for all $u,v \in V$, $u \neq v$, 
exactly one of the two pairs $(u,v)$ and $(v,u)$ belongs to $E$.

\item $G$ is {\em semicomplete} if for all $u,v \in V$, $u \neq v$, 
at least one of the two pairs $(u,v)$ and $(v,u)$ belongs to $E$.

\item $G$ is {\em (bidirectional) complete} if for all $u,v \in V$, $u \neq v$, 
both of the two pairs $(u,v)$ and $(v,u)$ belong to $E$.
\end{itemize}

\paragraph{Omitting the directions} 
For some given digraph $G=(V,E)$, we define
its {\em underlying undirected graph} by ignoring the directions of the 
edges, i.e.~$\un(G)=(V,\{ \{u,v\} \mid (u,v) \in E \text{ or } (v,u) \in E\})$.

\paragraph{Orientations} 
There are several ways to define a digraph $G=(V,E)$ from an undirected 
graph $G_u=(V,E_u)$.
If we replace every edge $\{u,v\}\in E_u$ by 
\begin{itemize}
\item
one of the arcs $(u,v)$ and $(v,u)$, we denote $G$ as an {\em orientation} of $G_u$.
Every digraph $G$  which can be obtained by an orientation of some undirected
graph $G_u$ is called an {\em oriented graph}.

\item
one or both of the arcs $(u,v)$ and $(v,u)$, we denote $G$ as a {\em biorientation} of $G_u$.
Every digraph $G$  which can be obtained by a biorientation of some undirected
graph $G_u$ is called a {\em bioriented graph}.

\item
both arcs $(u,v)$ and $(v,u)$, we denote $G$ as a {\em complete biorientation} of $G_u$.
Since in this case $G$ is well defined by $G_u$ we also denote
it by $\overleftrightarrow{G_u}$.
Every digraph $G$  which can be obtained by a complete biorientation of some undirected
graph $G_u$ is called a {\em complete bioriented graph}. 
\end{itemize}

\subsection{Recursively defined Digraphs}

\subsubsection{Operations}

The following operations
have already been considered  by Bechet et al.\ in \cite{BGR97,JRST01}.
Let $G_1=(V_1,E_1), \ldots, G_k=(V_k,E_k)$ be $k$ vertex-disjoint digraphs. 
\begin{itemize}
\item
The {\em disjoint union} of $G_1, \ldots, G_k$, 
denoted by $G_1 \oplus \ldots \oplus G_k$, 
is the digraph with vertex set $V_1\cup \ldots \cup V_k$ and 
arc set $E_1\cup \ldots \cup E_k$. 

\item
The {\em series composition} of $G_1,\ldots, G_k$, 
denoted by $G_1\otimes \ldots \otimes G_k$, 
is defined by their disjoint union plus all possible arcs between
vertices of $G_i$ and $G_j$ for all $1\leq i,j\leq k$, $i\neq j$.

\item
The {\em order composition} of $G_1, \ldots, G_k$, 
denoted by $G_1\oslash \ldots \oslash G_k$, 
is defined by their disjoint union plus all possible arcs from 
vertices of $G_i$ to vertices of $G_j$ for all $1\leq i < j\leq k$.

\item 
The {\em directed union} of  $G_1, \ldots, G_k$, 
denoted by $G_1\ominus \ldots \ominus G_k$, 
is defined by their disjoint union  plus  possible arcs from 
vertices of $G_i$ to vertices of $G_j$ for all $1\leq i < j\leq k$.\footnote{That is,
$G_1, \ldots, G_k$ are  induced subdigraphs of $G_1\ominus \ldots \ominus G_k$ and 
there is no edge $(u,v)$ in
$G_1\ominus \ldots \ominus G_k$
such that $v\in V_i$ and $u\in V_j$ for $j>i$.
The 
directed union generalizes the disjoint union and the order composition.}
\end{itemize}

\subsubsection{Directed co-graphs}

We recall the definition of directed co-graphs from \cite{CP06}. 

\begin{definition}[Directed co-graphs, \cite{CP06}]
The class of {\em directed co-graphs} is recursively defined as follows.
\begin{enumerate}[(i)]
\item Every digraph on a single vertex $(\{v\},\emptyset)$, 
denoted by $\bullet$, is a {\em directed co-graph}.

\item If  $G_1,\ldots,G_k$  are vertex-disjoint directed co-graphs, then 
\begin{enumerate}
\item 
the disjoint union  
$G_1\oplus\ldots \oplus G_k$,

\item 
the series composition
$G_1 \otimes \ldots \otimes G_k$, and 
\item
the order composition 
$G_1\oslash \ldots \oslash G_k$  are {\em directed co-graphs}.
\end{enumerate}
\end{enumerate}
\end{definition}

By the definition we conclude that for every directed co-graph $G=(V,E)$ the
underlying undirected graph $\un(G)$ is a co-graph, but not vice versa.

Similar as undirected co-graphs by the $P_4$, also directed co-graphs
can be characterized by excluding eight forbidden induced subdigraphs \cite{CP06}.

Obviously for every directed co-graph we can define a tree structure, 
denoted as {\em di-co-tree}. The leaves of the di-co-tree represent the 
vertices of the graph and the inner nodes of the di-co-tree  correspond 
to the operations applied on the subexpressions defined by the subtrees. 
For every directed co-graph one can construct a di-co-tree in linear time, 
see \cite{CP06}. 
The following lemma shows that it suffices to consider binary di-co-trees.

\begin{lemma}\label{le-di-co}
Every di-co-tree $T$ can be transformed into an equivalent binary di-co-tree $T'$, such
that every inner vertex in $T'$ has exactly two sons.
\end{lemma}

\begin{proof}
Let $G$ be a directed co-graph and $T$ be a di-co-tree for $G$. Since 
the disjoint union $\oplus$, the series composition $\otimes$, 
and the order composition $\oslash$ is associative, i.e. 
$G_1 \oplus\ldots \oplus G_k = (G_1 \oplus \ldots \oplus G_{k-1})\oplus G_k$,
we can transform $T$ recursively into a binary  di-co-tree $T'$ for $G$.
\end{proof}

Using the di-co-tree a lot of hard problems have been shown to be 
solvable in polynomial time when restricted to directed co-graphs \cite{Gur17a}.
In \cite{GWY16} the relation of directed co-graphs to the set of
graphs of directed NLC-width 1 and to the set of graphs of directed clique-width 2 
is analyzed. By \cite{HSW17,NEMWH18} directed co-graphs  are very useful for the reconstruction of the evolutionary 
history of genes or species using genomic sequence data.

\begin{lemma}\label{le-com}
Let $G$ be some digraph, then the following properties hold. 
\begin{enumerate}[1.]
\item
Digraph $G$ is a directed co-graph if and only if digraph $\overline{G}$ is  a directed co-graph.

\item 
Digraph $G$ is a directed co-graph if and only if digraph $G^c$ is  a directed co-graph.
\end{enumerate}
\end{lemma}

\subsubsection{Extended directed co-graphs}

Since the directed union generalizes the disjoint union and also the order composition
we can generalize the class of directed co-graphs as follows.

\begin{definition}[Extended directed co-graphs]
The class of {\em extended directed co-graphs} is recursively defined as follows.
\begin{enumerate}[(i)]
\item Every digraph on a single vertex $(\{v\},\emptyset)$, 
denoted by $\bullet$, is an {\em extended directed co-graph}.

\item If  $G_1,\ldots,G_k$  are vertex-disjoint extended directed co-graphs, then 
\begin{enumerate}
\item 
the directed union  
$G_1\ominus\ldots \ominus G_k$ and 

\item 
the series composition
$G_1 \otimes \ldots \otimes G_k$ are {\em extended directed co-graphs}.
\end{enumerate}
\end{enumerate}
\end{definition}

Also for every extended directed co-graph we can define a tree structure, 
denoted as {\em ex-di-co-tree}. The leaves of the ex-di-co-tree represent the 
vertices of the graph and the inner nodes of the ex-di-co-tree  correspond 
to the operations applied on the subexpressions defined by the subtrees.  
Following Lemma \ref{le-di-co} it suffices to consider binary ex-di-co-trees.

By applying the directed union which is not a  disjoint union and an order composition
we can obtain digraphs whose complement digraph is not an extended directed co-graph.
An example for this leads the directed path on $3$ vertices
$\overrightarrow{P_3}=(\{v_1,v_2,v_3\},\{ (v_1,v_2),(v_2,v_3)\})$.
Thus we only can carry over one of the two results shown in Lemma \ref{le-com}
to the class of extended directed co-graphs.

\begin{lemma}\label{le-com2}
Let $G$ be some digraph. 
Digraph $G$ is an extended directed co-graph if and only if digraph $G^c$ 
is  an  extended directed co-graph.
\end{lemma}

\section{Directed path-width}\label{pw-co}

According to Bar{\'a}t \cite{Bar06}, the notation of directed path-width was
introduced by Reed, Seymour, and Thomas around 1995 and relates to directed
tree-width introduced by Johnson, Robertson, Seymour, and Thomas in
\cite{JRST01}.

\begin{definition}[directed path-width]
A {\em directed path-decomposition} of a digraph $G=(V,E)$
is a sequence $(X_1, \ldots, X_r)$ of subsets of $V$, called {\em bags}, such 
that the following three conditions hold true.
\begin{enumerate}[(dpw-1)]
\item\label{dp1} $X_1 \cup \ldots \cup X_r ~=~ V$.
\item\label{dp2} For each $(u,v) \in E$ there is a pair $i \leq j$ such that
  $u \in X_i$ and $v \in X_j$.
\item\label{dp3}  If $u \in X_i$ and $u \in X_j$ for some  $u\in V$ and two 
indices $i,j$ with $i \leq j$, then $u \in X_\ell$ for all indices $\ell$ 
with $i \leq \ell \leq j$.
\end{enumerate}
The {\em width} of a directed path-decomposition ${\cal X}=(X_1, \ldots, X_r)$ 
is $$\max_{1 \leq i \leq r} |X_i|-1.$$ The {\em directed path-width} of $G$,
$\dpws(G)$ for short, is 
the smallest integer $w$ such that there is a directed path-de\-com\-po\-sition of
$G$ of width $w$. 
\end{definition}

\begin{lemma}[\cite{YC08}]\label{le-d-ud}
Let $G$ be some digraph, then $\dpws(G)\leq \pws(u(G))$.\footnote{The proofs
shown in \cite{YC08} use the notation of directed vertex separation number, which
is known to be equal to directed path-width.}
\end{lemma}

\begin{lemma}[\cite{Bar06}]\label{le-c-bi}
Let $G$ be some  complete bioriented digraph, then $\dpws(G)= \pws(u(G))$.
\end{lemma}

The proof can be done straightforward since a 
for $G$ of width $k$ leads to a layout for $\overleftrightarrow{G}$ 
of width at most $k$ and vice versa.

Determining whether the (undirected) path-width of some given (undirected) graph  is 
at most some given value $w$ is NP-complete 
\cite{KF79} 
even for bipartite graphs, complements
of bipartite graphs \cite{ACP87}, 
chordal graphs \cite{Gus93},
bipartite distance hereditary graphs \cite{KBMK93},  
and planar graphs with maximum vertex
degree 3 \cite{MS88}. 
Lemma \ref{le-c-bi} implies
that determining whether the directed path-width of some given digraph  is 
at most some given value $w$ is NP-complete even for digraphs whose underlying 
graphs lie in the mentioned classes.
On the other hand, determining whether the (undirected) path-width of some given (undirected) graph  is 
at most some given value $w$ is polynomial for permutation graphs \cite{BKK93},
circular arc graphs \cite{ST07a}, and co-graphs \cite{BM93}.

While undirected path-width can be solved by an FPT-algorithm \cite{Bod96}, 
the existence of such an algorithm for directed path-width is still open.
The directed path-width of a digraph $G=(V,E)$ can be computed in time 
$\bigo(\nicefrac{|E|\cdot |V|^{2\dpws(G)}}{(\dpws(G)-1)!})$ 
by \cite{KKKTT16}.
This leads to an XP-algorithm
for directed path-width w.r.t.~the standard parameter
and implies that for each constant $w$, it is decidable in polynomial time whether a given 
digraph has directed path-width at most $w$.

In order to prove our main results we show some properties
of directed path-decompositions. Similar results are known
for undirected path-decompositions and are useful within
several places.

\begin{lemma}[\cite{YC08}]\label{le-pw-subdigraph}
Let $G$ be some digraph and $H$ be an induced subdigraph
of $G$, then $\dpws(H)\leq \dpws(G)$.
\end{lemma}

\begin{lemma}[Bidirectional complete subdigraph]\label{cl1}
Let $G=(V,E)$ be some
digraph, $G'=(V',E')$ with $V'\subseteq V$ be a bidirectional complete subdigraph, 
and $(X_1, \ldots, X_r)$  a directed path-decomposition of $G$.
Then there is some $i$, $1\leq i\leq r$, such that $V'\subseteq X_i$.
\end{lemma}

\begin{proof}
We show the claim by an induction on $|V'|$. 
If $|V'|=1$ then by (dpw-1)
there is some $i$, $1\leq i\leq r$, such that $V'\subseteq X_i$. 
Next let $|V'|>1$ and $v\in V'$.
By our induction hypothesis there is some  $i$, $1\leq i\leq r$, such that $V'-\{v\}\subseteq X_i$. 
By (dpw-3)  there are two integers $r_1$ and $r_2$,
$1\leq r_1\leq r_2\leq r$, such that $v\in X_j$ for all $r_1\leq j \leq r_2$. 
If $r_1\leq i \leq r_2$ then $V'\subseteq X_i$. 
Next suppose that $i< r_1$ or $r_2<i$. If $i< r_1$ we define $j'=r_1$ and if 
$i> r_2$ we define $j'=r_2$.
We will show that $V'\subseteq X_{j'}$. 
Let $w\in V'-\{v\}$. Since there are two arcs $(v,w)$ and $(w,v)$ in $E$  
by (dpw-2) 
there is some  $r_1\leq j'' \leq r_2$ such that $v,w\in X_{j''}$. 
By (dpw-3) we conclude $w\in X_{j'}$. 
Thus $V'-\{v\}\subseteq X_{j'}$ and  $\{v\}\subseteq X_{j'}$, i.e. $V'\subseteq X_{j'}$.
\end{proof}

\begin{lemma}\label{cbsl} 
Let $G=(V,E)$ be a
digraph and $(X_1, \ldots, X_r)$  a di\-rec\-ted path-decomposition of $G$.
Further let $A,B\subseteq V$, $A\cap B=\emptyset$, 
and $\{(u,v),(v,u)~|~u\in A, v\in B\}\subseteq E$.
Then there is some $i$, $1\leq i\leq r$, such that $A\subseteq X_i$ or $B\subseteq X_i$.
\end{lemma}

\begin{proof}
Suppose that $B\not\subseteq X_i$ for all $1\leq i \leq r$.
Then there are $b_1,b_2\in B$ and $i_{1,\ell},i_{1,r},i_{2,\ell},i_{2,r}$, $1\leq i_{1,\ell}\leq i_{1,r}<i_{2,\ell}\leq i_{2,r}\leq r$,
such that $\{i~|~b_1\in X_i\}=\{i_{1,\ell},\ldots,i_{1,r}\}$ 
and  $\{i~|~b_2\in X_i\}=\{i_{2,\ell},\ldots,i_{2,r}\}$ (and both sets are disjoint).
Let $a\in A$. Since  $(b_2,a)\in E$ there is some $i_{2,\ell}\leq i \leq r$ such that $a\in X_i$
and since  $(a,b_1)\in E$ there is some $1\leq j\leq i_{1,r}$ such that $a\in X_j$.
By (dpw-3) it is true that $a\in X_k$ for every $ i_{1,r}\leq k \leq i_{2,\ell}$.

If we suppose  $A\not\subseteq X_i$ for all $1\leq i \leq r$ it follows that 
$b\in X_k$ for every $ i_{1,r}\leq k \leq i_{2,\ell}$. 
\end{proof}

\begin{lemma}\label{cl2}
Let  ${\cal X}=(X_1, \ldots, X_r)$ be a directed path-decomposition of 
some digraph $G=(V,E)$.
Further let $A,B\subseteq V$, $A\cap B=\emptyset$, and $\{(u,v),(v,u)~|~u\in A, v\in B\}\subseteq E$.
If there is some $i$, $1\leq i\leq r$, such that $A\subseteq X_i$ then there
are $1\leq i_1\leq i_2\leq r$ such that
\begin{enumerate}
\item for all $i$, $i_1\leq i \leq i_2$ is $A\subseteq X_i$,
\item $B\subseteq \cup_{i=i_1}^{i_2} X_i$, and
\item ${\cal X}'=(X'_{i_1},\ldots,X'_{i_2})$ where $X'_i=X_i\cap (A\cup B)$ 
is a directed path-decomposition of the digraph induced by $A\cup B$.
\end{enumerate}
\end{lemma}

\begin{proof}
Let $i_1=\min\{i~|~A\subseteq X_i\}$ and $i_2=\max\{i~|~A\subseteq X_i\}$. 
Since ${\cal X}$ satisfies (dpw-3),   it holds (1.).

Since there is some $i$, $1\leq i\leq r$, such that $A\subseteq X_i$ 
we know that ${\cal X}=(X_1, \ldots, X_r)$ is also a directed path-decomposition of 
$G'=(V,E')$, where $E'=E\cup\{(u,v)~|~u,v \in A, u\neq v\}$.
For every $b\in B$ the graph with vertex set $\{b\}\cup A$ is bidirectional complete subdigraph
of $G'$ which implies by  Lemma \ref{cl1} that there is some $i$, $i_1\leq i \leq  i_2$ such that 
$A\cup\{b\}\subseteq X_i$. Thus there is some $i$, $i_1\leq i \leq  i_2$ such that 
$b\in X_i$ which leads to (2.).

In order to show (3.) we observe that for the sequence ${\cal X}'=(X'_{i_1},\ldots,X'_{i_2})$
condition (dpw-1) holds by (1.) and (2.).

By (1.) and (2.) the arcs between two vertices from $A$ and the arcs between
a vertex from $A$ and a vertex from $B$ satisfy  (dpw-2). 
So let $(b',b'')\in E$ such that $b',b''\in B$.
By (2.) we know that  $b'\in X_i$ and $b''\in X_j$ for $i_1\leq i,j\leq i_2$.
If $j<i$ then by (dpw-3) for  ${\cal X}$ there is some $X_{j'}$,
$j'>i_2$ such that $b''\in X_{j'}$ but by (dpw-3) for ${\cal X}$ 
is $b''\in X_i$.

Further ${\cal X}'$ satisfies (dpw-3) since ${\cal X}$ satisfies (dpw-3).
\end{proof}

\begin{theorem}\label{th-pw}
Let $G=(V_G,E_G)$ and $H=(V_H,E_H)$ be two vertex-disjoint digraphs,
then the following properties hold. 
\begin{enumerate}[1.]
\item
$\dpws(G\oplus H)=\max\{\dpws(G),\dpws(H)\}$

\item 
$\dpws(G\oslash H)=\max\{\dpws(G),\dpws(H)\}$

\item 
$\dpws(G\ominus H)=\max\{\dpws(G),\dpws(H)\}$

\item 
$\dpws(G\otimes H)=\min\{\dpws(G)+|V_H|,\dpws(H)+|V_G|\}$
\end{enumerate}
\end{theorem}

\begin{proof}
\begin{enumerate}[1.] 
\item 
In order to show  $\dpws(G\oplus H)\leq \max\{\dpws(G),\dpws(H)\}$
we consider a directed path-decomposition $(X_1, \ldots, X_r)$ for $G$ and
a directed path-decomposition $(Y_1, \ldots, Y_s)$  for $H$. Then 
$(X_1, \ldots, X_r,Y_1, \ldots, Y_s)$ leads to a directed path-decomposition of $G\oplus H$.

Since $G$ and $H$ are induced subdigraphs of $G\oplus H$, by Lemma \ref{le-pw-subdigraph}
the directed path-width of both digraphs leads to a lower bound on the directed path-width
for the combined graph.

\item 
By the same arguments as used for (1.).
\item 
By the same arguments as used for (1.).

\item
In order to show  $\dpws(G\otimes H)\leq \dpws(G)+|V_H|$ 
let $(X_1, \ldots, X_r)$ be a directed path-decomposition of $G$. 
Then we
obtain by $(X_1\cup V_H, \ldots, X_r\cup V_H)$ a directed path-decomposition of $G\otimes H$.
In the same way a directed path-decomposition of $H$ leads to a directed path-decomposition of $G\otimes H$
which implies that $\dpws(G\otimes H)\leq \dpws(H)+|V_G|$.
Thus $\dpws(G\otimes H)\leq \min\{\dpws(G)+|V_H|,\dpws(H)+|V_G|\}$.

For the reverse direction let
${\cal X}=(X_1, \ldots, X_r)$  be a directed path-de\-com\-position of $G\otimes H$.
By Lemma \ref{cbsl} we know that there is some $i$, $1\leq i\leq r$, such that 
$V_G\subseteq X_i$ or $V_H\subseteq X_i$. We assume that $V_G\subseteq X_i$.
We apply Lemma \ref{cl2} using $G\otimes H$ as digraph, $A=V_G$ and $B=V_H$ in order
to obtain a directed path-decomposition  ${\cal X}'=(X'_{i_1},\ldots,X'_{i_2})$ for $G\otimes H$
where for all $i$, $i_1\leq i \leq i_2$, it holds $V_G\subseteq X_i$ and $V_H\subseteq \cup_{i=i_1}^{i_2} X_i$.
Further ${\cal X}''=(X''_{i_1},\ldots,X''_{i_2})$, where $X''_{i}= X'_{i}\cap V_H$
leads to a directed path-decomposition of $H$.
Thus there is some $i$, $i_1\leq i \leq i_2$, such that
$|X_i\cap V_H|\geq \dpws(H)+1$.  Since $V_G\subseteq X_i$, we know
that $|X_i\cap V_H|=|X_i|-|V_G|$ and thus
$|X_i|\geq  |V_G| + \dpws(H)+1$.
Thus the width of directed path-decomposition $(X_1, \ldots, X_r)$
is at least $\dpws(H)+|V_G|$.

If we assume that $V_H\subseteq X_i$ it follows that the
width of directed path-decomposition $(X_1, \ldots, X_r)$
is at least $\dpws(G)+|V_H|$.
\end{enumerate}
\end{proof}



\begin{lemma}\label{le-order}
 Let $G$ and $H$ be two directed co-graphs,
then $\pws(\un(G\oslash H))>\dpws(G\oslash H)$.
\end{lemma}

\begin{proof}
Let $G$ and $H$ be two directed co-graphs.
$$\begin{array}{lcl}
\pws(\un(G\oslash H)) & = & \pws(\un(G)\times\un(H)) \\
 &  = & \min\{\pws(\un(G))+|V_H|,\pws(\un(H))+|V_G|\}  ~~~~ \text{(by \cite{BM93})}\\
 &  > & \min\{\pws(\un(G))+ \dpws(H),\pws(\un(H))+ \dpws(G)\}  ~~~~ \\
 &  \geq & \min\{\dpws(G)+ \dpws(H),\dpws(H)+ \dpws(G)\}  ~~~~ \\
 &  = & \dpws(G)+ \dpws(H)\\
 &  \geq  & \max\{\dpws(G),\dpws(H)\}\\
 &  = & \dpws(G\oslash H)
\end{array}
$$
\end{proof}

\begin{corollary}\label{xx}
Let $G$ be some directed co-graph, then $\dpws(G)=\pws(u(G))$ if and only if there is
an expression for $G$ without any order operation. 
Further $\dpws(G)=0$ if and only if there is
an expression for $G$ without any series operation.
\end{corollary}

\begin{proof}
If there is a  construction without
order operation, then Theorem \ref{th-pw} and
the results of \cite{BM93} imply $\dpws(G)=\pws(u(G))$. If there is a  construction using an
order operation, Lemma \ref{le-order} implies that $\dpws(G)\neq\pws(u(G))$.
\end{proof}

\section{Directed tree-width}\label{tw-co}

An {\em acyclic} digraph ({\em DAG} for short) is a digraph without any 
cycles as subdigraph. An  out-tree  is a digraph with a distinguished root such that
all arcs are directed away from the root. For two vertices $u,v$ of an out-tree $T$
the notation $u\leq v$ means that there is a directed path on $\geq 0$ 
arcs from $u$ to $v$ and $u < v$ means that there is a directed path on $\geq 1$ 
arcs from $u$ to $v$.

Let $G=(V,E)$ be some digraph and $Z\subseteq V$. A vertex set $S\subseteq V$
is {\em $Z$-normal}, if there is no directed walk in $G-Z$ with first and last vertices in
$S$ that uses a vertex of $G-(Z\cup S)$.  That is, a set $S\subseteq V$ is $Z$-normal,
if every directed walk which leaves and again enters  $S$ must contain
only vertices from $Z\cup S$. Or,  a set $S\subseteq V$ is $Z$-normal, if every directed walk which leaves and
again enters $S$ must contain a vertex from $Z$.


\begin{definition}[directed tree-width, \cite{JRST01}]\label{D3}
A {\em (arboreal) tree-decomposition} of a digraph $G=(V_G,E_G)$ is a triple  $(T, \mathcal{X}, \mathcal{W})$. 
Here $T=(V_T,E_T)$ is an out-tree, 
$\mathcal{X}=\{X_e ~|~ e\in E_T\}$ and
$\mathcal{W}=\{W_r~|~ r \in V_T\}$ are  sets of subsets of $V_G$, such that the following
two conditions hold true.
\begin{enumerate}[(dtw-1)]
\item $\mathcal{W}=\{W_r~|~ r \in V_T\}$ is a partition of $V_G$ into nonempty 
subsets.\footnote{A remarkable difference to the undirected tree-width \cite{RS86} is that the sets
$W_r$ have to be disjoint and non-empty.}
\item For every $(u,v)\in E_T$ the set $\bigcup\{W_r ~|~ r\in V_T, v\leq r\}$ is 
$X_{(u,v)}$-normal. 
\end{enumerate}
The {\em width} of a  (arboreal) tree-decomposition $(T, \mathcal{X}, \mathcal{W})$ is 
$$\max_{r\in V_T} |W_r\cup \bigcup_{e \sim r} X_e|-1.$$
Here $e \sim r$ means that $r$ is one of the two vertices of arc $e$.
The {\em directed tree-width} of $G$, $\dtws(G)$ for short, is the smallest
integer $k$ such that there is a (arboreal) tree-decomposition $(T, \mathcal{X}, \mathcal{W})$ of $G$
of width $k$. 
\end{definition}

\begin{remark}[$Z$-normality]
Please note that our definition of $Z$-normality slightly  differs  from the following
definition in \cite{JRST01} where $S$ and $Z$ are disjoint. 
A vertex set $S\subseteq V-Z$ is {\em $Z$-normal}, 
if there is no directed walk in $G-Z$ with first and last vertices in
$S$ that uses a vertex of $G-(Z\cup S)$.  
That is, a set $S\subseteq V-Z$ is $Z$-normal,
if every directed walk in $G-Z$ which leaves and again enters  $S$ must contain
only vertices from $Z\cup S$.
Or, a set $S\subseteq V-Z$ is $Z$-normal, if every directed walk which leaves and
again enters $S$ must contain a vertex from $Z$, see \cite{BG09}.

Every set $S\subseteq V-Z$  which is is $Z$-normal w.r.t.~the definition in \cite{JRST01}
is also $Z$-normal w.r.t.~our definition. Further a set  $S\subseteq V$ 
which is $Z$-normal w.r.t.~our definition, is also $Z-S$-normal
w.r.t.~the definition in \cite{JRST01}.
Thus the directed tree-width of a digraph is equal for both definitions of $Z$-normality.
\end{remark}

\begin{lemma}[\cite{JRST01}]\label{le-tw-d-ud}
Let $G$ be some digraph, then $\dtw(G)\leq \tws(\un(G))$.
\end{lemma}

\begin{lemma}[\cite{JRST01}]\label{le-tw-c-bi}
Let $G$ be some  complete bioriented digraph, then $\dtw(G)= \tws(\un(G))$.
\end{lemma}


Determining whether the (undirected) tree-width of some given (undirected) graph  is 
at most some given value $w$ is NP-complete 
even for bipartite graphs and complements
of bipartite graphs \cite{ACP87}.
Lemma \ref{le-tw-c-bi} implies
that determining whether the directed tree-width of some given digraph  is 
at most some given value $w$ is NP-complete even for digraphs whose underlying 
graphs lie in the mentioned classes.

The results of \cite{JRST01} lead to an XP-algorithm
for directed tree-width w.r.t.~the standard parameter
which implies that for each constant $w$, it is decidable in polynomial time whether a given
digraph has directed tree-width at most $w$.

In order to show our main results we show some properties
of directed tree-decompositions.

\begin{lemma}[\cite{JRST01}]\label{le-tw-subdigraph2}
Let $G$ be some digraph and $H$ be an induced subdigraph
of $G$, then $\dtws(H)\leq \dtws(G)$.
\end{lemma}

\begin{lemma}[Bidirectional complete subdigraph]\label{cl1tw}
Let $(T, \mathcal{X}, \mathcal{W})$, $T=(V_T,E_T)$, where $r_T$ is the root of $T$, 
be a directed tree-decomposition of some digraph $G=(V,E)$ and 
$G'=(V',E')$ with $V'\subseteq V$ be a bidirectional complete subdigraph.
Then $V'\subseteq W_{r_T}$ 
or there is some $(r,s)\in E_T$, such that $V'\subseteq W_{s}\cup X_{(r,s)}$.
\end{lemma}

\begin{proof}
First we show the existence of a vertex $s$ in $V_T$,
such that $W_s\cap V'\neq \emptyset$ but for every vertex $s'$ such that $s<s'$
 holds  $W_{s'}\cap V'= \emptyset$. If there is a leaf $\ell$ in $T$, such that 
$W_{\ell}\cap V' \neq \emptyset$, we can choose $s=\ell$. Otherwise
we look for vertex $s$ among the predecessors of the leaves in $T$, and so on.
Since $V'\subseteq V=\cup_{r\in V_T} W_r$ we will find a vertex $s$ with the 
stated properties.

Next we show that $W_s$ leads to a set which shows the statement of the lemma.
If $s$ is the root of $T$, then $W_{s'}\cap V'\neq \emptyset$ for none of its 
successors $s'$ in $T$ i.e. $W_{s'}\cap V'=\emptyset$ for all of its successors 
$s'$ in $T$, which implies by (dtw-1) that $V'\subseteq W_s$.
Otherwise let $r$ be the predecessor of $s$ in $T$. If $V'\subseteq W_s$
the statement is true. Otherwise let $c\in V'-W_s$ and $c'\in V'\cap W_s$.
Then $(c,c')\in E$ and $(c',c)\in E$ implies that $c\in X_{(r,s)}$
by (dtw-2).
\end{proof}

\begin{lemma}\label{cbsltw} 
Let $G=(V,E)$ be some
digraph,  $(T, \mathcal{X}, \mathcal{W})$, $T=(V_T,E_T)$, where $r_T$ is the root of $T$,  
be a directed tree-decomposition of $G$.
Further let $A,B\subseteq V$, $A\cap B=\emptyset$, 
and $\{(u,v),(v,u)~|~u\in A, v\in B\}\subseteq E$.
Then $A\cup B\subseteq W_{r_T}$ 
or there is some $(r,s)\in E_T$, such that  $A\subseteq W_{s}\cup X_{(r,s)}$
or $B\subseteq W_{s}\cup X_{(r,s)}$.
\end{lemma}

\begin{proof}
Similar as in the proof of Lemma \ref{cl1tw} we can 
find a vertex $s$ in $V_T$,
such that $W_s\cap (A\cup B)\neq \emptyset$ but for every vertex $s'$ such that $s<s'$
holds  $W_{s'}\cap (A\cup B)= \emptyset$. 

If $s$ is the root of $T$, then $W_{s'}\cap (A\cup B)\neq \emptyset$ for none 
of its successors $s'$ in $T$, i.e. $W_{s'}\cap (A\cup B)=\emptyset$ for all 
of its successors $s'$ in $T$, which implies by (dtw-1) that $A\cup B\subseteq W_s$.

Otherwise let $r$ be the predecessor of $s$ in $T$.
If $A\cup B\subseteq W_s$ the statement is true. Otherwise we know that 
there is some $a\in A$ such that $a\in W_s$ and $B\not\subseteq W_s$ 
or some $b\in B$ such that  $b\in W_s$ and $A\not\subseteq W_s$.
We assume that there is some $a\in A$ such that $a\in W_s$ and $B\not\subseteq W_s$. 
Let $b\in B-W_s$ and $a\in A\cap W_s$.
Then $(a,b)\in E$ and $(b,a)\in E$  implies that $b\in X_{(r,s)}$
by (dtw-2). Thus we have shown $B\subseteq W_{s}\cup X_{(r,s)}$.

If we assume that there some $b\in B$ such that  $b\in W_s$, we conclude 
$A\subseteq W_{s}\cup X_{(r,s)}$. 
\end{proof}

\begin{lemma}\label{w-one}
Let $G$ be a digraph of directed tree-width at most $k$.
Then there is a directed tree-decomposition $(T, \mathcal{X}, \mathcal{W})$, $T=(V_T,E_T)$, 
of width at most $k$ for $G$ such that $|W_r|=1$ for every $r\in V_T$.
\end{lemma}

\begin{proof}
Let  $G=(V,E)$ be a digraph and  $(T, \mathcal{X}, \mathcal{W})$, $T=(V_T,E_T)$, 
be a directed tree-decomposition of $G$. Let $r\in V_T$ such that $W_r=\{v_1,\ldots,v_k\}$
for some $k>1$. Further let $p$ be the predecessor of $r$ in $T$ and $s_1,\ldots,s_\ell$
be the successors of  $r$ in $T$.
Let $(T', \mathcal{X}', \mathcal{W}')$ be defined by the following modifications of $(T, \mathcal{X}, \mathcal{W})$:
We replace vertex $r$ in $T$ by the directed path
$P(r)=(\{r_1,\ldots,r_k\},\{(r_1,r_2),\ldots,(r_{k-1},r_k)\})$
and replace arc $(p,r)$ by $(p,r_1)$ and the $\ell$ arcs $(r,s_j)$, $1\leq j\leq \ell$,
by the $\ell$ arcs $(r_k,s_j)$, $1\leq j\leq \ell$ in $T'$.
We define the sets $W'_{r_j}=\{v_{j}\}$ for $1\leq j \leq k$. Further we 
define the sets 
$X'_{(p,r_1)}=X_{(p,r)}$, $X_{(r_k,s_j)}=X_{(r,s_j)}$, $1\leq j\leq \ell$,
and $X'_{(r_j,r_{j+1})}=X_{(p,r)}\cup\{r_1,\ldots,r_j\}$,  $1\leq j\leq k-1$.

By our definition $\mathcal{W}'$ leads to a
partition of $V$ into nonempty subsets. 
Further for every new arc $(r_{i-1},r_i)$, $1 < i \leq k$,
the set $\bigcup\{W'_{r'} ~|~ r'\in V_{T'}, r_i\leq r'\}$ is 
$X'_{(r_{i-1},r_i)}$-normal 
since $\bigcup\{W_{r'} ~|~ r'\in V_T, r\leq r'\}$ is 
$X_{(p,r)}$-normal and $X'_{(r_{i-1},r_i)}=X_{(p,r)}\cup \{r_1,\ldots,r_{i-1}\}$. 
The property is fulfilled for arc $(p,r_1)$ and $(v_k,s_j)$, $1\leq j \leq \ell$
since the considered vertex sets of $G$ did not change.
Thus triple $(T', \mathcal{X}', \mathcal{W}')$ is
a directed tree-decomposition of $G$.

The width of  $(T', \mathcal{X}', \mathcal{W}')$
is at most the width of $(T, \mathcal{X}, \mathcal{W})$ since 
for every $r_j$, $1\leq j\leq k$, the following holds:   
$|W'_{r_j}\cup \bigcup_{e \sim r_j} X'_e|\leq |W_{r}\cup \bigcup_{e \sim r} X_e|$.

If we perform this transformation for every $r\in V_T$ such that $|W_r|>1$, 
we obtain a  directed tree-decomposition of $G$ which fulfills the
properties of the lemma.
\end{proof}

\begin{lemma}\label{le-dist}
Let $G=(V,E)$ be a digraph of directed tree-width at most $k$,
such that $V_1\cup V_2=V$, $V_1\cap V_2=\emptyset$, and $\{(u,v),(v,u)~|~u\in V_1, v\in V_2\}\subseteq E$.
Then there is a directed tree-decomposition $(T, \mathcal{X}, \mathcal{W})$, $T=(V_T,E_T)$, 
of width at most $k$ for $G$ such that for every $e\in E_T$ holds $V_1\subseteq X_e$
or for every $e\in E_T$ holds $V_2\subseteq X_e$.
\end{lemma}

\begin{proof}
Let $G=(V,E)$ be a digraph of directed tree-width at most $k$ and  
$(T, \mathcal{X}, \mathcal{W})$, $T=(V_T,E_T)$,
be a directed tree-decomposition of width at most $k$ for $G$.
By Lemma \ref{w-one} we can assume that holds: $|W_r|=1$ for every $r\in V_T$.

We show the claim by traversing $T$ in a bottom-up order.
Let $t'$ be a leaf of $T$, $t$ be the predecessor of $t'$ in $T$ 
and $W_{t'}=\{v\}$ for some  $v\in V_1$. Then the following holds: 
$V_2\subseteq X_{(t,t')}$ since $(v,v')\in E$ and $(v',v)\in E$ 
for every $v'\in V_2$.

If $t'$ is a non-leaf of $T$ and there is a successor $t''$ of $t'$ in $T$ such that
$V_1\subseteq X_{(t',t'')}$ and there is a successor $t'''$ of $t'$ in $T$ such that
$V_2\subseteq X_{(t',t''')}$. Then the width of  $(T, \mathcal{X}, \mathcal{W})$
is $|V_1|+|V_2|-1$ which allows us to insert $V_1$ into every set $X_e$ as well as  
$V_2$ into every set $X_e$.

Otherwise let $t'$ be a non-leaf of $T$ and $V_2\subseteq X_{(t',t'')}$ for 
every successor $t''$ of $t'$.
Let $t$ be the predecessor of $t'$ and $s$ be the predecessor of $t$ in $T$. 
We distinguish the following two cases.
\begin{itemize}
\item Let $V_1\subseteq \cup_{t'\leq \tilde{t}} W_{\tilde{t}}$. We replace $X_{(t,t')}$ by $X_{(t,t')} \cup V_2$
in order to meet our 
claim for edge $(t,t')$.

We have to show that this does not increase the width of the obtained directed
tree-decomposition at vertex $t'$ and at vertex $t$.

The value of $|W_{t'}\cup \bigcup_{e \sim t'} X_e|$ does not change, since $V_2\subseteq X_{(t',t'')}$
by induction hypothesis and $(t',t'')\sim t'$.

Since $V_1\subseteq \cup_{t\leq \tilde{t}} W_{\tilde{t}}$ by (dtw-2) we can assume that
$V_1\cap X_{(s,t)}=\emptyset$. 
Since all $W_r$ have size one 
we know that $|W_{t}\cup \bigcup_{e \sim t} X_e|\leq |W_{t'}\cup \bigcup_{e \sim t'} X_e|$.

\item Let $V_1\not\subseteq \cup_{t'\leq \tilde{t}} W_{\tilde{t}}$. We distinguish the following two cases.
\begin{itemize}
\item Let $V_2\cap \cup_{t'\leq \tilde{t}} W_{\tilde{t}}=\emptyset$, then $W_{t'}=\{v\}$ for some  $v\in V_1$
and thus $V_2\subseteq X_{(t,t')}$ since $(v,v')\in E$ and $(v',v)\in E$ 
for every $v'\in V_2$.

\item  Let $V_2\cap \cup_{t'\leq \tilde{t}} W_{\tilde{t}}\not=\emptyset$.
Since 
$\{(u,v),(v,u)~|~u\in V_1, v\in V_2\}\subseteq E$ the following is true:
\begin{equation}
V-\bigcup_{t'\leq \tilde{t}} W_{\tilde{t}}= (V_1\cup V_2)-\bigcup_{t'\leq \tilde{t}} W_{\tilde{t}} \subseteq X_{(t,t')}.\label{eq}
\end{equation}
That is, all vertices of $G$ which are not of one of the sets
$W_{\tilde{t}}$ for all successors $\tilde{t}$ of $t'$ are in set $X_{(t,t')}$.

We define $X_{(t,t')}=(V - \cup_{t'\leq  \tilde{t}} W_{\tilde{t}}) \cup V_2$ in order to meet our 
claim for edge $(t,t')$.

We have to show that this does not increase the width of the obtained directed
tree-decomposition at vertex $t'$ and and vertex $t$.

The value of $|W_{t'}\cup \bigcup_{e \sim t'} X_e|$ does not change, since $V_2\subseteq X_{(t',t'')}$
by induction hypothesis and $(t',t'')\sim t'$ and by (\ref{eq}).

Further (\ref{eq}) implies that  $X_{(s,t)}\subseteq X_{(t,t')}$ and thus
$|W_{t}\cup \bigcup_{e \sim t} X_e|\leq |W_{t'}\cup \bigcup_{e \sim t'} X_e|$.  
\end{itemize}
\end{itemize}

Thus if $T$ has a leaf $t'$ such that $W_{t'}=\{v\}$ for some  $v\in V_1$
we obtain a  directed tree-decomposition $(T, \mathcal{X}, \mathcal{W})$, $T=(V_T,E_T)$,
such that $V_2\subseteq X_e$ for every $e\in E_T$. 
And if $T$ has a leaf $t'$ such that $W_{t'}=\{v\}$ for some  $v\in V_2$
we obtain a  directed tree-decomposition $(T, \mathcal{X}, \mathcal{W})$, $T=(V_T,E_T)$,
such that  $V_1\subseteq X_e$ for every $e\in E_T$.  
\end{proof}

\begin{theorem}\label{th-johnson}
Let $G=(V_G,E_G)$ and $H=(V_H,E_H)$ be two vertex-disjoint digraphs,
then the following properties hold. 
\begin{enumerate}
\item
$\dtws(G\oplus H)=\max\{\dtws(G),\dtws(H)\}$

\item 
$\dtws(G\oslash H)=\max\{\dtws(G),\dtws(H)\}$

\item 
$\dtws(G\ominus H)=\max\{\dtws(G),\dtws(H)\}$

\item 
$\dtws(G\otimes H)=\min\{\dtws(G)+|V_H|,\dtws(H)+|V_G|\}$

\end{enumerate}
\end{theorem}

\begin{proof}
Let $G=(V_G,E_G)$ and $H=(V_H,E_H)$ be two vertex-disjoint digraphs.
Further let $(T_G, \mathcal{X}_G, \mathcal{W}_G)$ be a directed tree-decomposition of $G$
such that  $r_G$ is the root of $T_G=(V_{T_G},E_{T_G})$ and $(T_H, \mathcal{X}_H, \mathcal{W}_H)$
be a directed tree-decomposition of $H$
such that  $r_H$ is the root of $T_H=(V_{T_H},E_{T_H})$. 

\begin{enumerate}
\item
We define  a directed tree-decomposition $(T_J, \mathcal{X}_J, \mathcal{W}_J)$ for $J=G\oplus H$.
Let $\ell_G$ be a leaf of $T_G$. Let $T_J$ be the disjoint union of $T_G$ and $T_H$ 
with an additional arc  $(\ell_G,r_H)$.
Further let 
$\mathcal{X}_J=\mathcal{X}_G \cup  \mathcal{X}_H\cup \{X_{(\ell_G,r_H)}\}$, where $X_{(\ell_G,r_H)}=\emptyset$  
and $\mathcal{W}_J=\mathcal{W}_G\cup \mathcal{W}_H$.
Triple $(T_J, \mathcal{X}_J, \mathcal{W}_J)$ satisfies (dtw-1) since 
the combined decompositions satisfy  (dtw-1). Further  $(T_J, \mathcal{X}_J, \mathcal{W}_J)$ satisfies (dtw-2)
since additionally in $J$ there is no arc from a vertex of $H$ to a vertex of $G$. This shows that 
$\dtws(G\oplus H)\leq\max\{\dtws(G),\dtws(H)\}$.
Since $G$ and $H$ are induced subdigraphs of $G\oplus H$, by Lemma \ref{le-tw-subdigraph2}
the directed tree-width of both leads to a lower bound on the directed tree-width
for the combined graph.

\item
The same arguments lead to $\dtws(G\oslash H)=\max\{\dtws(G),\dtws(H)\}$.

\item
The same arguments lead to $\dtws(G\ominus H)=\max\{\dtws(G),\dtws(H)\}$.

\item 
In order to show  $\dtws(G\otimes H)\leq \dtws(G)+|V_H|$ 
let $T_J$ be the disjoint union of a new root $r_J$ and $T_G$ 
with an additional arc  $(r_J,r_G)$. Further let 
$\mathcal{X}_J=\mathcal{X}'_G \cup  \{X_{(r_J,r_G)}\}$, where $\mathcal{X}'_G =\{X_e\cup V_H~|~e\in E_{T_G}\}$ 
and $X_{(r_J,r_G)}=V_H$  
and $\mathcal{W}_J=\mathcal{W}_G\cup \{W_{r_H}\}$, where $W_{r_J}=V_H$.
Then we obtain by  $(T_J, \mathcal{X}_J, \mathcal{W}_J)$  
a directed tree-decomposition of width at most $ \dtws(G)+|V_H|$ for $G\otimes H$.

In the same way a new root $r_J$ and $T_H$ 
with an additional arc  $(r_J,r_H)$, $\mathcal{X}'_H =\{X_e\cup V_G~|~e\in E_{T_H}\}$, 
$X_{(r_J,r_H)}=V_G$,  $W_{r_J}=V_G$
lead to a directed tree-decomposition of width at most $ \dtws(H)+|V_G|$ for $G\otimes H$.
Thus $\dtws(G\otimes H)\leq \min\{\dtws(G)+|V_H|,\dtws(H)+|V_G|\}$.

For the reverse direction let $(T_J, \mathcal{X}_J, \mathcal{W}_J)$, $T_J=(V_T,E_T)$,
be a directed tree-de\-com\-position of minimal width for $G\otimes H$.
By Lemma \ref{le-dist} we can assume that $V_G\subseteq X_e$ for every $e\in E_T$
or $V_H\subseteq X_e$ for every $e\in E_T$. 
Further by Lemma \ref{w-one} we can assume that $|W_t|=1$ for every $t\in V_T$.

We assume that $V_G\subseteq X_e$ for every $e\in E_T$.
We define $(T'_J, \mathcal{X}'_J, \mathcal{W}'_J)$,  $T'_J=(V'_T,E'_T)$, by $X'_{e}= X_{e}\cap V_H$
and $W'_s=W_s \cap V_H$. Whenever this leads to an empty set $W'_s$ where $t$
is the predecessor of $s$ in $T'_J$ we remove
vertex $s$ from $T'_J$ and replace every arc $(s,t')$ by $(t,t')$ with
the corresponding set $X_{(t,t')}=X_{(s,t')}\cap V_H$.

Then $(T'_J, \mathcal{X}'_J, \mathcal{W}'_J)$ is a  
directed tree-decomposition of $H$  as follows.  
\begin{itemize}
\item $\mathcal{W}'_J$ is a partition of $V_H$ into nonempty sets.

\item Let $e$ be an arc in $T'_J$ which is also in $T_J$. Since $e \sim s$ implies
$W_s=W'_s=\{v\}$ for some $v\in V_H$ normality condition remains true.

Arcs $(t,t')$ in $T'_J$ which are not in $T_J$  are obtained by two arcs $(t,s)$  and
$(s,t')$ from $T_J$.  If $\cup\{W_r ~|~ r\in V_T, t'\leq r\}$ is 
$X_{(s,t')}$-normal, then $\cup\{W_r ~|~ r\in V'_T, t'\leq r\}$ is 
$X_{(t,t')}$-normal since $X_{(t,t')}=X_{(s,t')}\cap V_H$.
\end{itemize}

The width of  $(T'_J, \mathcal{X}'_J, \mathcal{W}'_J)$  is at most $\dtws(G\otimes H)-|V_G|$ as follows.  
\begin{itemize}
\item Let $s$ be a vertex in $T'_J$ such that $W_t\cap V_H\neq \emptyset$ for all $(s,t)$ in $T_J$.
$$\begin{array}{lcll}
|W'_{s}\cup \bigcup_{e \sim s} X'_e| &= &|(W_{s}\cap V_H)\cup \bigcup_{e \sim s} (X_e\cap V_H)|& \text{ by definition} \\
       &=& |(W_{s}\cup \bigcup_{e \sim s} X_e)\cap V_H| & \text{ factor out } V_H \\
       &=& |W_{s}\cup \bigcup_{e \sim s} X_e|-  |V_G|  & \text{ since }  V_G\subseteq X_e \\
\end{array} 
$$
\item Let $s$ be a vertex in $T'_J$ such that there is $(s,t)$ in $T_J$ with $W_t\cap V_H= \emptyset$.
%
%
%
\begin{eqnarray}
|W'_{s}\cup \bigcup_{e \sim s} X'_e| &=&
|(W_{s}\cap V_H)\cup  \left( X_{(t'',s)}  \cap V_H \right)
\cup\bigcup_{\substack{(s,t) \in E_T \\ W_t \cap V_H=\emptyset}} \left( X_{(t,t')} \cap V_H \right)\nonumber \\
&&\cup \bigcup_{\substack{(s,t) \in E_T \\ W_t \cap V_H\neq \emptyset}} \left( X_{(s,t)}  \cap V_H \right)| \label{no-eq}
\end{eqnarray}

In order to bound this value we observe that for $W_t\cap V_H= \emptyset$ the following is true:
$W_t=\{v\}$ for $v\in V_G$. Then $X_{(s,t)}=((V_G\cup V_H)-\cup_{t\leq \tilde{t} } W_{\tilde{t}}) \cup V_G$ by 
Lemma \ref{le-dist}.
%
%
%
That is, $X_{(s,t)}$ consists
of all vertices from $V_G$ and all vertices which are not of one of the
sets $W_{\tilde{t}}$ for all successors $\tilde{t}$ of $t$.
Applying this argument to $X_{(t,t')}$ 
we only can have $v$ as an additional vertex. 
But since $v\in V_G$ we know that $v\in X_{(s,t)}$ by our assumption.
This implies 
\begin{equation}
X_{(t,t')}\subseteq X_{(s,t)} \text{ for all arcs } (s,t) \text{ in } T_J \text{ such that } W_t\cap V_H= \emptyset \label{rel}
\end{equation}
which allows the following estimations:
$$\begin{array}{lcll}
|W'_{s}\cup \bigcup_{e \sim s} X'_e| &= &|(W_{s}\cap V_H)\cup \bigcup_{e \sim s} (X_e\cap V_H)|  & \text{ by } (\ref{no-eq}) \text{ and } (\ref{rel})\\
      &=  &|(W_{s} \cup \bigcup_{e \sim s} X_e)\cap V_H| & \text{ factor out } V_H\\
      &=  & |W_{s}\cup \bigcup_{e \sim s} X_e|-  |V_G|    & \text{ since }  V_G\subseteq X_e \\
\end{array} 
$$
\end{itemize}

Thus the width of $(T'_J, \mathcal{X}'_J, \mathcal{W}'_J)$  is at most $\dtws(G\otimes H)-|V_G|$
and since $(T'_J, \mathcal{X}'_J, \mathcal{W}'_J)$ is a directed tree-decomposition of $H$
it follows $\dtws(H)\leq \dtws(G\otimes H)-|V_G|$ 

If we assume that $V_H\subseteq X_e$ for every $e\in E_T$ 
it follows that $\dtws(G)\leq \dtws(G\otimes H)-|V_H|$. 
\end{enumerate}
\end{proof}

The proof of Theorem \ref{th-johnson} even shows that for any directed co-graph
there is a tree-decomposition $(T, \mathcal{X}, \mathcal{W})$  
of minimal width such that $T$ is a path.



Similar to the path-width results, we conclude the following results.

\begin{lemma}\label{le-order-tw}
 Let $G$ and $H$ be two directed co-graphs,
then $\tws(\un(G\oslash H))>\dtws(G\oslash H)$.
\end{lemma}

\begin{corollary}\label{xxtw}
Let $G$ be some directed co-graph, then $\dtws(G)=\tws(u(G))$ if and only if there is
an expression for $G$ without any order operation. 
Further $\dtws(G)=0$ if and only if there is
an expression for $G$ without any series operation.
\end{corollary}

\section{Directed tree-width and directed path-width of special digraphs}\label{pw-tw-cogr}

For general digraphs the directed tree-width is at most the directed path-width are
by the following Lemma.

\begin{lemma}\label{le-tw-pw}
Let $G$ be some digraph, then $\dtws(G)\leq \dpws(G)$.
\end{lemma}

\begin{proof}
Let  ${\cal Y}=(Y_1,\ldots,Y_r)$ be a directed path-decomposition of
some digraph $G$. We obtain a directed tree-decomposition of
$G$ by $(T, \mathcal{X}, \mathcal{W})$, where $$T=(\{v_1,\ldots,v_r\},\{(v_1,v_2),\ldots, (v_{r-1},v_r)\})$$
is a directed path, $W_{v_1}=Y_1$, $W_{v_i}=Y_i-(Y_1\cup\ldots Y_{i-1})$, and $X_{(v_i,v_{i+1})}=Y_i\cap Y_{i+1}$.
Since $W_i\cup X_{(v_{i-1},v_i)}\cup X_{(v_{i},v_i+1)}\subseteq Y_i$ it follows
that $\dtws(G)\leq \dpws(G)$.
\end{proof}

There are several examples where the equality does not
hold. 

\begin{example}
Every complete biorientation of a rooted tree has directed tree-width $1$ and
a directed path-width depending on its height. 
The path-width of perfect $2$-ary trees of hight $h$ is 
$\lceil h/2\rceil$ (cf. \cite{Sch89}) and for $k\geq 3$ the  
path-width of  perfect $k$-ary trees of hight $h$ is exactly $h$ 
by Corollary 3.1 of \cite{EST94}.
\end{example}

\subsection{Directed Co-graphs}

\begin{theorem}\label{main}
For every directed co-graph $G$, it holds that $\dpws(G)=\dtws(G)$.
\end{theorem}

\begin{proof}
Let $G=(V,E)$ be some  directed co-graph. We show the result by
induction on the number of vertices $|V|$. 
If $|V|=1$, then $\dpws(G)=\dtws(G)=0$. If $G=G_1 \oplus G_2$, then by
Theorem \ref{th-pw} and Theorem \ref{th-johnson} follows:
$$\dpws(G)=\max\{\dpws(G_1),\dpws(G_2)\}=\max\{\dtws(G_1),\dtws(G_2)\}=\dtws(G).$$
For the other two operations a similar relation holds.
\end{proof}

By Lemma \ref{le-c-bi} and Lemma  \ref{le-tw-c-bi}
our results generalize the known results from \cite{BM90,BM93}
but can not be obtained by the known results.

\begin{theorem} For every directed co-graph $G=(V,E)$ which is given by 
a binary di-co-tree the directed path-width and directed tree-width can be
computed in time $\bigo(|V|)$.
\end{theorem}

\begin{proof}
The statement follows by the algorithm given in Fig.\ \ref{fig:algo}, 
Theorem \ref{th-pw}, and Theorem \ref{th-johnson}. The necessary
sizes of the subdigraphs defined by subtrees of di-co-tree $T_G$ 
can be precomputed in time $\bigo(|V|)$. 
\end{proof}

\begin{figure}[htbp]
\hrule
{\strut\footnotesize \bf Algorithm {\sc Directed Path-width}($v$)} 
\hrule
\vspace{-2mm}
\begin{tabbing}
xx \= xx \= xx \= xx \= xx \= xx \= xx\= xx \= xxxxxxxxxxxxxxxxxxxxxxx \=\kill
if $v$ is a leaf of di-co-tree $T_G$  \\
\> then $\dpws(G[T_v])=0$  \\
\> else \{  \\
\> \> Directed Path-width($v_\ell$) \>\>\>\>\>\>\> $\blacktriangleright$ \textit{$v_{\ell}$ is the left successor of $v$} \\
\> \> Directed Path-width($v_r$)    \>\>\>\>\>\>\> $\blacktriangleright$ \textit{$v_{r}$ is the right successor of $v$} \\
\> \> if $v$ corresponds to a $\oplus$, or a  $\oslash$  operation \\
\> \> \> then $\dpws(G[T_v])=\max\{\dpws(G[T_{v_\ell}]),\dpws(G[T_{v_r}])\}$   \\
\> \> \> else  $\dpws(G[T_v])=\min\{\dpws(G[T_{v_\ell}])+ |V_{G[T_{v_r}]}|,\dpws(G[T_{v_r}])+ |V_{G[T_{v_\ell}]}|\}$  \\
\> \}
\end{tabbing}
\vspace{-2mm}
\hrule
\caption{Computing the directed path-width of $G$ for every vertex of a di-co-tree $T_G$.}
\label{fig:algo}
\end{figure}

For general digraphs $\dpws(G)$ leads to a lower bound for $\pws(\un(G))$
and  $\dtws(G)$ leads to a lower bound for $\tws(\un(G))$, see \cite{Bar06,JRST01}.
For directed co-graphs we obtain a closer relation as follows.

\begin{corollary}\label{le-order2}
Let $G$ be a directed co-graph and $\overleftrightarrow{\omega}(G)$ be the size of
a largest bioriented clique of $G$. It then holds that
$$\overleftrightarrow{\omega}(G)=\dpws(G)-1=\dtws(G)-1\leq \pws(\un(G))-1=\tws(\un(G))-1=\omega(\un(G)).$$
All values are equal if and only if $G$ is a complete bioriented digraph.
\end{corollary}

\begin{proof}
The equality $\pws(\un(G))-1=\tws(\un(G))-1=\omega(\un(G))$
has been shown in \cite{BM90,BM93}. The equality
$\overleftrightarrow{\omega}(G)=\dpws(G)-1=\dtws(G)-1$ follows by 
Lemma \ref{cl1} (or Lemma \ref{cl1tw}) and Theorem \ref{main}. The upper bound
follows by Lemma \ref{le-d-ud} or Lemma \ref{le-tw-d-ud}. 
\end{proof}

\subsection{Extended Directed Co-graphs}

Theorem \ref{main} can be generalized to extended directed co-graphs.

\begin{theorem}\label{main-2}
For every extended directed co-graph $G$, it holds that $\dpws(G)=\dtws(G)$.
\end{theorem}

The algorithm shown in Fig.\ \ref{fig:algo} can be adapted  to show the
following result.

\begin{theorem}\label{th-general}
For every extended directed co-graph $G=(V,E)$ which is given by 
a binary ex-di-co-tree  the directed path-width and directed tree-width can be
computed in time $\bigo(|V|)$.
\end{theorem}

In order to process the strong components of a digraph we recall the
following definition. The {\em acyclic condensation} of a digraph $G$, $AC(G)$ for short,
is the digraph whose vertices are the strongly connected
components $V_1,\ldots,V_c$ of $G$ and there is an edge from
$V_i$ to $V_j$ if there is an edge $(v_i,v_j)$ in $G$ such 
that $v_i\in V_i$ and $v_j\in V_j$. Obviously for every
digraph $G$ the digraph  $AC(G)$ is always acyclic.

\begin{lemma}\label{le-repr-sc}
Every digraph $G$ can be represented by the directed union of its strong components.
\end{lemma}

\begin{proof}
Let $G$ be a digraph,
$AC(G)$ be the acyclic condensation of $G$, and $v_1,\ldots,v_c$ be a 
topological ordering of $AC(G)$, i.e. for every edge $(v_i,v_j)$ in $AC(G)$
it holds $i<j$. Further let $V_1,\ldots,V_c$ be the vertex sets of its strong components
ordered by the topological ordering. Then
$G$ can be obtained by $G=G[V_1] \ominus \ldots  \ominus G[V_c]$.
\end{proof}

\begin{theorem}\label{th-strongc}
Let $G$ be a digraph, then it holds:
\begin{enumerate}
\item The directed tree-width of $G$ is the maximum tree-width of its strong components.
\item The directed path-width of $G$ is the maximum path-width of its strong components.
\end{enumerate}
\end{theorem}

\begin{proof}
Follows by Lemma \ref{le-repr-sc} and Theorem \ref{th-pw} and Theorem \ref{th-johnson}.
\end{proof}

The directed path-width result of Theorem \ref{th-strongc}
was also shown in \cite{YC08} using the directed vertex separation number,
which is equal to the directed path-width.

\section{Conclusion and Outlook}\label{sec-concl}

In this paper we could generalize the equivalence
of path-width and tree-width of co-graphs which is
known from \cite{BM90,BM93} to directed graphs.
The shown equality also holds for more general directed tree-width 
definitions such as allowing empty sets $W_r$ in \cite{JRST01a}. 

This is not possible for the directed tree-width approach suggested by Reed in  \cite{Ree99},
which uses  sets $W_r$ of size one only for the leaves of $T$ of a directed 
tree-decomposition $(T, \mathcal{X}, \mathcal{W})$. 
To obtain a counter-example
let $S_{1,n}=(V,E)$ be a star graph on $1+n$ vertices, i.e. $V=\{v_0,v_1,\ldots,v_n\}$
and $E=\{\{v_0,v_i\}~|~1\leq i \leq n\}$. Further let $G_n$ be the complete biorientation
of $S_{1,n}$, which is a directed co-graph. 
Then $\tws(S_{1,n})=1$ and by Theorem \ref{main} and Theorem \ref{le-tw-d-ud}
we know $\dpws(G_n)=\dtws(G_n)\leq 1$. 
Using the approach of  \cite{Ree99} in any possible tree-decomposition $(T, \mathcal{X}, \mathcal{W})$ 
for $G_n$ there is a leaf $u$ of $T$ such that $W_u=\{v_0\}$. Further there is some $u'\in V_T$,
such that  $(u',u)\in E_T$. By normality for edge $(u',u)$ it holds $X_{(u',u)}=\{v_1,\ldots,v_n\}$
which implies that using the approach of \cite{Ree99} the directed tree-width of $G$ is at least $n$.

The approach given in
\cite[Chapter 6]{DE14} using strong components within
(dtw-2) should be considered in future work.
Further research directions should extend
the shown results to larger classes as well as
consider related width parameters.

The class of directed co-graphs was studied very well in \cite{CP06}.  For the
class of extended directed co-graphs it remains to show how to compute an
ex-di-co-tree in order to apply Theorem \ref{th-general}. 

\section*{Acknowledgements}

The work of the second author was supported by the German Research 
Association (DFG) grant  GU 970/7-1.



\end{document}